\documentclass{llncs}

\pdfoutput=1

\usepackage{theorem}
\usepackage{graphicx}
\usepackage{amssymb}
\usepackage{amsmath}
\usepackage{algpseudocode}
\usepackage[chapter]{algorithm}
\usepackage{color}

\spnewtheorem*{theorem_app}{Theorem}{\bfseries}{\itshape}
\begin{document}
\title{Online Exploration of Polygons with Holes}

\author{Robert Georges
        \and
        Frank Hoffmann
        \and
        Klaus Kriegel
        } 
        
\institute{Freie Universit\"at Berlin, Institut f\"ur Informatik, 14195 Berlin, Germany\\
\email{\{georges, hoffmann, kriegel\}@mi.fu-berlin.de}}

\maketitle

\begin{abstract} 
We study  online strategies for autonomous  mobile robots with vision to explore unknown polygons with at most $h$ holes. Our main contribution is an $(h+c_0)!$--competitive strategy for such polygons under the assumption that each hole is marked with a special color, where $c_0$ is a universal constant. The strategy is based on a new hybrid approach. Furthermore, we give a new lower bound construction for small $h$.
\end{abstract} 

\section{Introduction}

A classical basic task \cite{B,PY} for an autonomous mobile robot is to explore an unknown environment modeled by a polygon, possibly with polygonal holes. We assume the robot to be point shaped and to start from a given point, $s$, on the polygon's outer  boundary. It is equipped with an unlimited $360^{\circ}$ vision system that continuously provides the visibility polygon of its current position.
When the robot has observed every point of the polygon it returns to $s$.

Considering a known polygon, an optimal tour $\mathcal{T}_{\mathrm{opt}}$ through $s$ can be computed offline.
The robot's performance exploring the unknown polygon online is evaluated through competitive analysis. Therefore we compare the length of the tour generated by the robot with the length of $\mathcal{T}_{\mathrm{opt}}$.
If this ratio is bounded from above by a constant $\mathcal{C}$ for any problem instance, we call the strategy $\mathcal{C}$--competitive.

Over the last two decades, the problem of designing competitive online exploration strategies for certain polygon classes has received a lot of attention. A simple greedy strategy is almost optimal for simple orthogonal polygons as shown in  a seminal paper by Deng et al. \cite{DKP}, see also \cite{HNP}. Later, Hoffmann et al. \cite{HIKK} came up with a 26.5--competitive strategy for general
simple polygons (in the following called HIKK--strategy). On the other hand, there is a lower bound for the competitive ratio of 1.28 in this case \cite{HIL}. If one allows polygons with $h$ holes there is a
lower bound of $\Omega (\sqrt{h})$, even for orthogonal polygons \cite{AKS}, and computing the optimal offline tour becomes NP--hard.

 The only positive result in the presence of holes we are aware of is an $O(h)$--competitive strategy  for orthogonal polygons with $h$ holes \cite{DKP}. This result yields a 14--competitive strategy ($L_1$--metric) for the case of  one hole \cite{G}.

Surprisingly, there are no competitive strategies known for general polygons with at most $h$ holes, even in the case $h=1$. Such strategies were conjectured to exist for each $h$ in \cite{DKP-FOCS}. One of the main differences between exploring orthogonal and
general polygons is the following. An optimal tour that learns  a single hole in an orthogonal polygon can always afford to encircle the hole. In contrast, in a general polygon the hole could have the shape of a thin long triangle and the path length needed to learn it is not necessarily related to its perimeter. Such a hole could be learned from a distance with minimal effort.

We make the following contributions to the problem of exploring polygons with  holes. In Section 2 we give for $h=1$ a rather simple lower bound of 2 for the orthogonal case and a lower bound of 2.618 for the competitive ratio in the general case. This latter bound also holds for a modified model, where the hole is specially colored and the robot can therefore distinguish between outer boundary edges and edges of the hole. Undoubtedly, this should be of great advantage
for the robot to fulfill its task. Nevertheless, it seems to be  nontrivial to come up with a competitive strategy  under this assumption. Subsequently we describe our strategy $h$--CPEX, which stands for {\it Colored Polygon EXploration}, and prove it to have a competitive factor that depends on $h$ only.

We start with describing  strategy $1$--CPEX in Section 3. It proceeds in two phases. In Phase 1 it follows the HIKK--strategy until the hole $H$ is eventually visible for the first time. Then it learns, based on a doubling strategy, the shortest tour $R$ encircling $H \cup \left\{s\right\}$.  In Phase 2 a novel hybrid approach is implemented to explore the remaining "caves" inside and outside of $R$. It is based on the knowledge of the length $\left|R\right|$. As soon as our strategy knows that $\left|R\right|$ is less than $c \cdot \left|\mathcal{T}_{\mathrm{opt}}\right|$, for a suitable universal constant $c$,  the hole is classified {\it safe}, meaning that we can encircle it without loosing competitiveness. To this end  we connect the hole with $s$ by introducing a barrier and invoke  the HIKK--strategy (Lemma \ref{hybrid}) for the modified polygon. Otherwise, the hole $H$ has the status {\it critical}. In this case, we subdivide the polygon by building a  fence line $f$ that connects the farthest (wrt. $s$) point of $H$ with the outer boundary. We get two simple polygons, the {\it front yard} {$\mathcal F$} containing $s$ and the {\it backyard} {$\mathcal B$}. Again, the front yard is explored using HIKK but as soon as the path exceeds a certain length bound, we interrupt since $H$ becomes safe and we proceed as before. Otherwise, we are left with the task of exploring the backyard. This is done by doubling arguments.

In Section 4 we generalize these ideas in a straightforward way to an arbitrary number $h$ of colored holes. Again, $h$--CPEX first uses HIKK until the first hole is found and classifies discovered holes $H$ to be safe or critical afterwards. As before the decisions are based on the knowledge of the length of the shortest tour encircling $H \cup \left\{s\right\}$.
In the case of a safe hole we can invoke a recursive call of $(h-1)$--CPEX. Even the strategy for critical holes can be adopted. A main difference is the use of a generalized doubling strategy, which is known as $m$--star search \cite{PY,K}.

We show that, due to its recursive structure, the competitive ratio  $\mathcal C_h$ of $h$--CPEX is bounded by $(h+c_0)!$ for a universal constant $c_0$. We have not tried to optimize it, our main goal is to show that it is bounded in $h$.
A closer look at the case $h=1$ yields a competitive factor of $\approx 610$, see \cite{G}.\\
We assume that the reader is familiar with  the HIKK-strategy \cite{HIKK}. It serves as the base case $0$--CPEX for the recursive part of the $h$--CPEX. Recall that $\mathcal C_0=26.5$. 

% ---------------------------------------------------------------------------------------------------
% Lower bounds
% ---------------------------------------------------------------------------------------------------

\section{Lower Bounds}

\begin{theorem}
\label{bounds}
Any deterministic online strategy $\mathcal{S}_1$ that computes valid watchman routes in polygons with at most one hole has a competitive ratio (1)  $\geq 2$ in the orthogonal case and (2) $\geq \frac{3+\sqrt{5}}{2}\approx 2.618$ in the general case.
\end{theorem}

\begin{proof}
The proof for the first part of Theorem \ref{bounds} can be found in Appendix \ref{app_orth_bound}.
For the proof of the second part consider the polygon given in Fig.\ref{lower_bound_2}(a). If the robot discovers a new reflex vertex $v$ that has an invisible incident edge, we call the extension of that edge into the polygon's interior the \textit{cut} of $v$. The cut of a reflex vertex on the outer boundary could either hit the hole or pass by on the left/right side. This leads to the following lower bound construction.

\begin{figure}[ht]
\centering
\includegraphics[scale=.59]{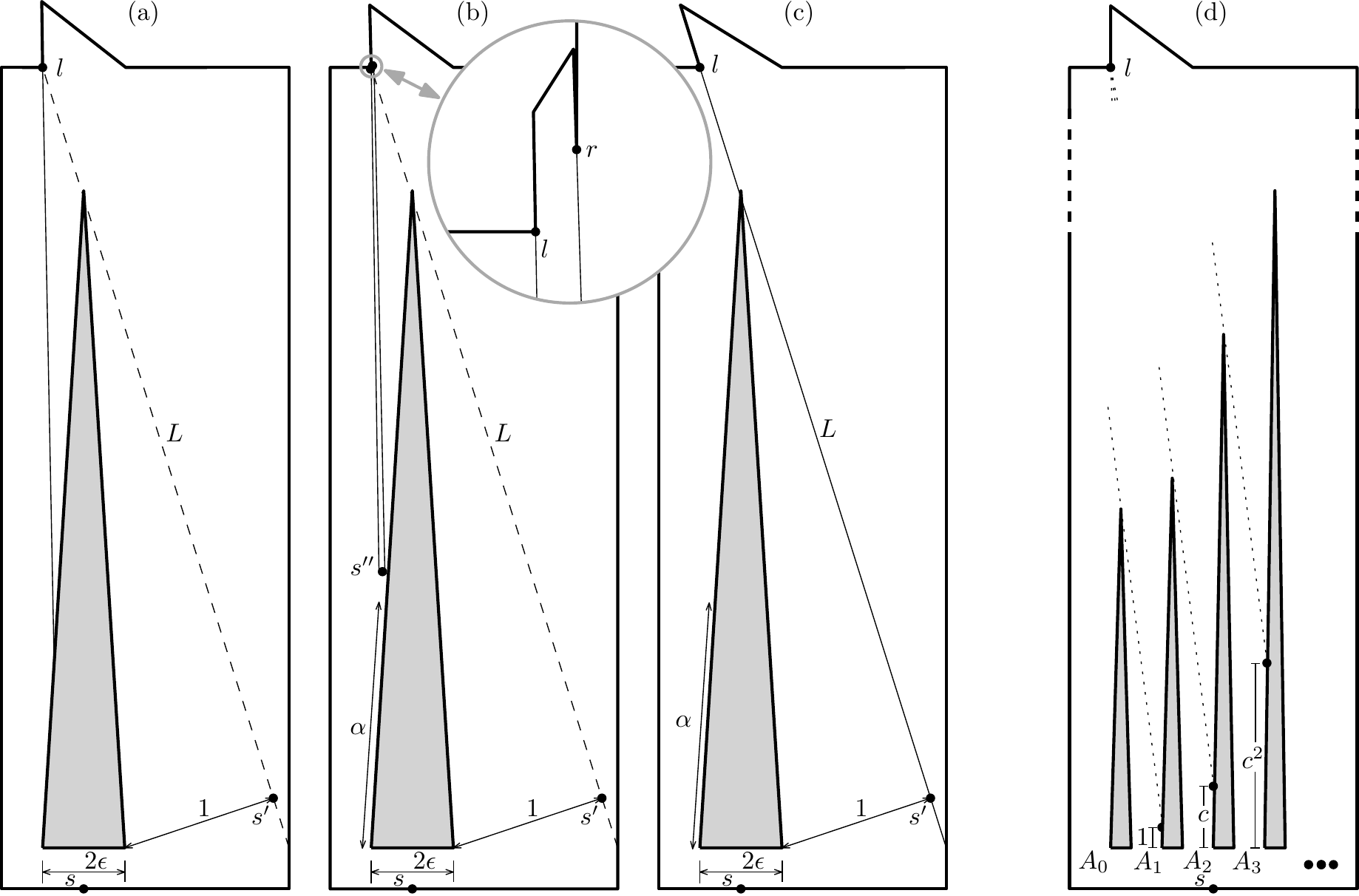}
\caption{Lower bound example, general case and extension to $h$ holes}
\label{lower_bound_2}
\end{figure}

 After traveling a distance of $3\epsilon$ an online strategy $\mathcal{S}_1$ has learned the hole completely and has discovered reflex vertex $l$.  
 Now, $\mathcal{S}_1$ knows line $L$ and point $s'$ on $L$, which is the closest point (say in distance 1) on the right side of the hole, where $\mathcal{S}_1$ could possibly learn the hidden edge behind vertex $l$. In fact, after reaching $s'$ the strategy perhaps sees everything and returns to $s$, Fig.\ref{lower_bound_2}(c).  
But directly moving to $s'$ is a fatal decision given the slightly modified situation in Fig.\ref{lower_bound_2}(b). Here it suffices to move distance $\alpha$ on the left side of the hole to point $s''$. There $\mathcal{S}_1$ learns both vertex $l$ and some reflex vertex $r$ that is hidden behind $l$. (Hint: $r$ is very close to $l$ and can be learned only from the left side in a competitive way.)
 
Any competitive strategy $\mathcal{S}_1$ must be able to handle both possibilities, the optimal strategy chooses the correct side. Therefore, $\mathcal{S}_1$ must try to explore  corner $l$ on the left side first, traveling some distance $\alpha$. Because of the malicious adversary, it still  misses the cut of $l$ by a very small distance. Then it will return and explore the cut of $l$ from the right side. 
If the task is completed in $s'$, the strategy travels a total distance of $2\alpha+2$.
But close to $s'$, it could also learn the existence of vertex $r$ and $\mathcal{S}_1$ has to return once again to the left side of the hole. This yields a total path length of $4\alpha +2$.

In both cases the quotient of the  tour length generated by $\mathcal{S}_1$ and the optimal tour length is a function in $\alpha$. 
The monotonically decreasing function $f(\alpha)=\frac{4\alpha +2}{2\alpha}$ describes the competitive ratio, if the cut points to the left side of the hole, Fig.\ref{lower_bound_2}(b). If the cut of $l$ is learned in $s'$ we have the monotonically increasing function $g(\alpha)=\frac{2\alpha+2}{2}$, Fig.\ref{lower_bound_2}(c).
Comparing both functions to determine the optimal value for $\alpha$ results in $\alpha=\frac{1+\sqrt{5}}{2}$, the golden ratio. We obtain $\frac{3+\sqrt{5}}{2}\approx 2.618$ as lower bound for the competitive ratio.
\qed
\end{proof}

In this context a colored hole would make no difference. The lower bound of $\Omega(\sqrt{h})$ for polygons with $h$ holes \cite{AKS} holds for the colored case, too. The problem of learning a reflex vertex in presence of vision-blocking holes turns out to be a fundamental issue for any strategy that wants to explore arbitrary polygons. 

\bigskip
\noindent
\textbf{Remark:} 
The lower bound construction in the general case can be extended to $h>1$ holes (indicated in Fig.\ref{lower_bound_2}(d)). For $h=2$ we get the lower bound of $\approx 2.9$, $h=3$ results in $\approx 3.02$. This does not lead to new results for $h\geq4$, because for increasing $h$ the obtained lower bound is decreasing again.

% ---------------------------------------------------------------------------------------------------
%  1-CPEX
% ---------------------------------------------------------------------------------------------------

\section{1--CPEX: Polygons with One Colored Hole}
As usual, we assume that the starting point $s$ is on the outer boundary of $\mathcal P$ and
${\mathcal T}_{\mathrm{opt}}$ denotes a shortest closed watchman tour. 
The design of our strategy follows the basic principle that in each phase and sub phase
the generated  path should have a length that compares to  $\left|{\mathcal T}_{\mathrm{opt}}\right|$ in a competitive way. 

\subsection{How to Explore a Bicolored Corridor}
A basic task during the exploration is to learn the structure of the hole. Because of the coloring, each edge of the polygon can be easily associated with the hole or the outer boundary. This motivates the problem of exploring a bicolored corridor, which can be seen as a natural extension of the Cow-Path problem, see \cite{CKM,PY}.

The corridor may have several branchings. The task is to find a target $t$, that sees walls of both colors (Fig.\ref{hull}(a)). This additional constraint guarantees, that $t$ cannot be located in an unicolored part of the polygon. At any time, the visibility polygon contains only two edges connecting two walls of different colors. These two so-called \textit{main windows} arise from two vision blocking vertices, which can be explored on a semicircle, see \cite{HIKK}. The doubling approach \cite{CKM,K} is used to link both exploration directions.

\begin{figure}[ht]
\centering
\includegraphics[scale=1]{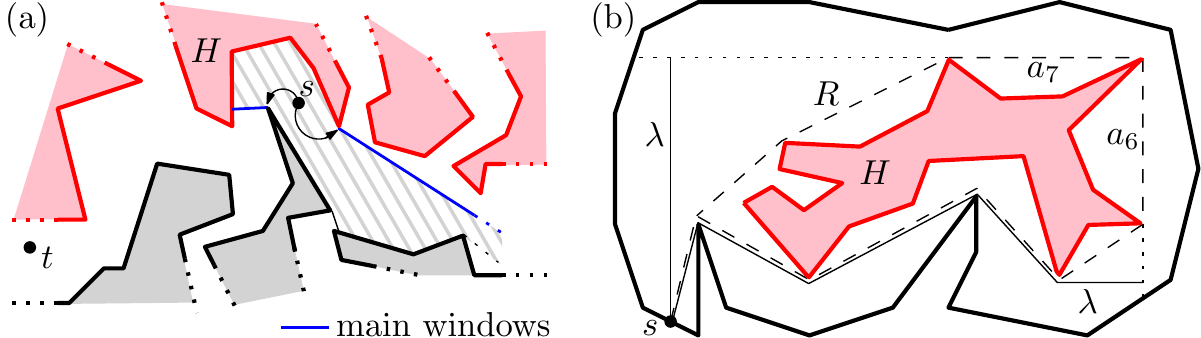}
\caption{(a) The bicolored corridor problem (b) Learning the shortest path $R$ encircling the hole}
\label{hull}
\end{figure}

\subsection{Phase 1: Learning the Shortest Tour Encircling the Hole}
\label{RCH}
If no point of the hole is visible from $s$, we start the $26.5$-competitive HIKK-strategy ($0$--CPEX) until $H$ becomes visible.  The next goal is  to look once around the hole, more precisely to learn the shortest tour $R$ around it. 
$R$ equals the boundary of the relative convex hull of $H \cup \{s\}$.\footnote{A set $M$ is relatively convex in   $\mathcal P$ if for each pair of points in $M$ the geodesics (shortest paths) connecting them are included in $M$.} Thus, one can imagine $R= \partial\left(\mathrm{RCH}\left(H \cup \{s\}\right)\right)$ as the shape of a rubber band spanned around $H$ and the starting point inside  $\mathcal P$, see Fig.\ref{hull}(b).

Any strategy that tries to learn $R$ circling $H$ in a fixed orientation will fail to be competitive. Consider e.g. the situation in Fig.\ref{lower_bound_2}(a). A strategy that explores $R$ in cw--orientation has to walk up to the top vertex of $H$ on the left side and down again on the right side of $H$. This can exceed $c\cdot \left|\mathcal{T}_{\mathrm{opt}}\right|$ for any constant $c$.

 Thus, the situation resembles the bicolored corridor problem. We explore $R$ in rounds via doubling, approaching the vertices corresponding to the main windows alternately on semicircles:
In an odd/even round $k$ we move in cw--/ccw--orientation $2^{k-1}$ length units. In each round there is a last known segment of $R$ corresponding to a part of the bicolored corridor. It is ending at a reflex vertex that is associated to the main window and hides the next segment.

Combining this with the factor $2$ of the semicircle strategy and the factor $9$ of the
doubling approach we can show that our strategy learns $R$ with total path length
$\leq 36 \left|{\mathcal T}_{\mathrm{opt}}\right|$.

\bigskip

\noindent
After learning $R$ we can derive the following lower bound $\lambda$ on $\left|\mathcal T_{\mathrm{opt}}\right|$.

Let $a_1,a_2, \ldots , a_n$ be the ccw-oriented chain of line segments defining $R$, starting from $s$. Any strategy that learns $R$ has to see each vertex $p_i \in R$, incident with $a_i$ and $a_{i+1}$, both from the right half-plane of $a_i$ and from the right half-plane of $a_{i+1}$. The path length to fulfill this task for $p_i$, maximized over all vertices of $R$, defines a lower bound $\lambda$ to learn $R$ and therefore a lower bound on $\left|\mathcal T_{\mathrm{opt}}\right|.$\footnote{This definition of $\lambda$ is equivalent with that given in \cite{GHK} for $h=1$ and it easily generalizes to the case $h>1$.} 

In Fig.\ref{hull}(b) the lower bound $\lambda$ is realized by the effort to learn $(a_6, a_7)$.

\subsection{ Phase 2: The Hybrid Approach}
\label{sec_hybrid}

The hole $H$ in $\mathcal P$ is called \textit{$c$--safe} (for a fixed constant $c$), if
$\left|R\right| \leq c  \left|{\mathcal T}_{\mathrm{opt}}\right|$
holds. As long as we don't know whether the hole is $c$--safe, the hole is called \textit{$c$--critical}. 
\\
\\
\textbf{Observation:} $|R| \leq c \lambda$ implies a $c$--safe hole. 
\\
\\
The hybrid approach consists in implementing the following rule: As soon 
as we learn that the hole is $c$--safe, our strategy will switch to the simple polygon mode using the following lemma.

\begin{lemma}\label{hybrid}
Any polygon  $\mathcal P$ with a $c$--safe hole $H$ can be explored 
with total path length $\leq (4c+2)\cdot{\mathcal C_0} \cdot \left|{\mathcal T}_{\mathrm{opt}}\right|$.
\end{lemma}

\begin{proof}
Consider a shortest path $b$ from $s$ to $H$. If it contains reflex vertices
of  $\mathcal P$, we slightly shift $b$ into the polygon's interior. This way we treat $b$ as an additional barrier that transforms  $\mathcal P$ into a simple polygon  $\mathcal P'$ (Fig.\ref{phase2}(a)). We show that $0$--CPEX (i.e., HIKK) applied
to $\mathcal P'$ fulfills the required condition. Therefore it is sufficient to show
that there is a closed watchman tour $\mathcal T_\mathrm{e}$ of length  $\left|\mathcal T_\mathrm{e}\right| \leq (4c+2) \cdot \left|{\mathcal T}_{\mathrm{opt}}\right|$ in $\mathcal P'$. 

\begin{figure}[ht]
\centering
\includegraphics[scale=1]{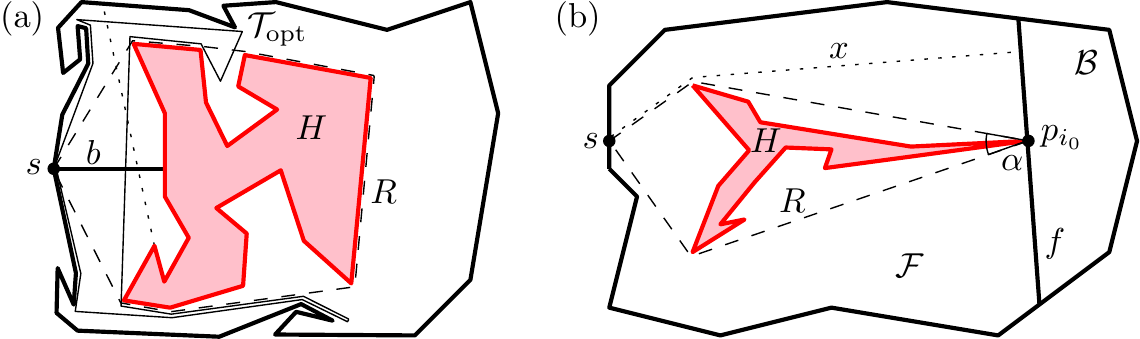}
\caption{Exploring a polygon with (a) a $c$--safe and (b) a $c$--critical hole.}
\label{phase2}
\end{figure}

Barrier $b$ can cut the original ${\mathcal T}_{\mathrm{opt}}$ into several (left and right) pieces. The visibility can be restricted, too. We use $R$ plus two copies of $b$ (one on the left, the other one on the right side of the barrier) to link together all pieces of ${\mathcal T}_{\mathrm{opt}}$ and to restore full vision. Doubling this structure we get an Eulerian graph and an Eulerian tour $\mathcal T_\mathrm{e}$. Finally:
\begin{equation*}
\left|\mathcal T_\mathrm{e}\right| = 2\left(|R|+2|b|+\left|{\mathcal T}_{\mathrm{opt}}\right|\right) \leq 4|R| +2 \left|{\mathcal T}_{\mathrm{opt}}\right| \leq (4c+2) \left|{\mathcal T}_{\mathrm{opt}}\right|  \enspace .
\end{equation*}
\qed
\end{proof}

Next we discuss how to proceed if the current status of the hole  is   $c$--critical. Assume that
the lower bound $\lambda$ was established by  edge pair $(a_{i_0},a_{i_0+1})$ with the
common polygon vertex $p_{i_0}$. Using elementary trigonometric reasoning one can  show that
the angle $\alpha$ between  $a_{i_0}$ and $ a_{i_0+1}$ is small. (Eventually we choose  $c =5$. This gives $\alpha < \frac{\pi}{6}$.) 
We define a fence $f$, subdividing $\mathcal P$ into two simple polygons ${\mathcal F}$ (the \textit{front yard}), and  ${\mathcal B}$ (the \textit{backyard}), see Fig.\ref{phase2}(b). $f$ is chosen to be the line segment perpendicular to the angular bisector of $\alpha$ through $p_{i_0}$.

\begin{lemma}\label{P0}
Let $H$ be a $c$--critical hole and $x$ be twice the shortest path length from $s$ to $f$ in $\mathcal P$. We have: $\mathcal F$ can be learned with tour length $\leq \mathcal C_0 \cdot x$ or $H$ is $c$--safe.
\end{lemma}

\begin{proof}(Sketch)
We invoke $0$--CPEX for $\mathcal F$ starting in $s$. If $\mathcal F$ gets explored with total tour length $\leq \mathcal C_0 \cdot x$, we are left with the task to explore $\mathcal B$. Otherwise we can prove that $H$ is $c$--safe. This follows from a simple case distinction. If ${\mathcal T}_{\mathrm{opt}}$ for $\mathcal P$ touches the fence, the claim is obvious. Otherwise  ${\mathcal T}_{\mathrm{opt}}$ is a tour inside  ${\mathcal F}$ and the claim follows from the competitiveness of  $0$--CPEX for $\mathcal F$. (Details in Appendix \ref{app_frontyard}) 
\qed
\end{proof}

If the hole becomes $c$--safe, we proceed as described in Lemma \ref{hybrid}.
In the other case it remains to explore the backyard ${\mathcal B}$ from $s$. Observe, we are now in a situation similar to our lower bound construction. We know there are reflex vertices in $\mathcal B$ that hide polygon edges we have to learn. But it is not clear whether to approach the corresponding cuts on the left or right side of the hole.

We describe how to learn a group of left reflex vertices, compare \cite{HIKK}. The existence of these vertices has been "observed" along the way while learning $R$, respectively $\mathcal F$. But, of course, this has not influenced the tours generated in these subroutines. 
Basically, cuts of such vertices can lie completely in $\mathcal B$ or they can cross the fence line $f$. As soon as we know that there is a cut not crossing $f$, $H$ becomes $c$--safe, since $\mathcal{T}_{\mathrm{opt}}$ intersects $f$.

\bigskip
A target vertex $l$ can be located in three different regions of $\mathcal B$ (Fig.\ref{backyard}). Cuts of vertices in $\mathcal B_3$ crossing $f$ on the left of $p_{i_0}$ have been explored along the way (as soon as they have been discovered), otherwise they cannot be visible yet. All other cuts of vertices in $\mathcal B_1$ and $\mathcal B_2$ crossing $f$ on the left are also crossing $R$, for the same reason. Therefore following the angle hull \cite{HIKK} of $R$ on the left side is a suitable way to explore these cuts (Details can be found in the Appendix \ref{app_background}). 

\begin{figure}[ht]
\centering
\includegraphics[scale=.9]{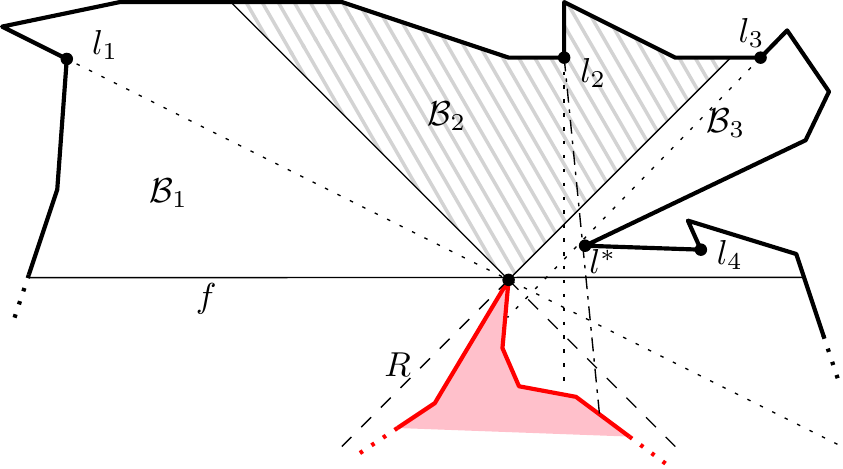}
\caption{Different types of left vertices in $\mathcal B$}
\label{backyard}
\end{figure}

\bigskip
A similar result can be shown for cuts crossing $f$ on the right part. If we are sure that a vertex has to be explored from the right, we put it into a special list $\mathcal V$ which is learned afterwards (only from the right side, without doubling).

Vertices in $\mathcal B_1$ are hidden behind the top vertex $p_{i_0}$ of $H$ ($l_1$ in \emph{Fig.}\ref{backyard}), otherwise, they would have been explored already. Therefore approaching $p_{i_0}$ on a semicircle is competitive, wherever the target vertex is located.
If it becomes visible, it is explored or it can be added to $\mathcal V$.

Vertices in $\mathcal B_2$ and $\mathcal B_3$ have either been seen from the right side and can be added to $\mathcal V$, too ($l_4$ in Fig.\ref{backyard}), or they have been discovered from the left and are hidden behind a vertex $l^*$ from the right ($l_2$ and $l_3$ in Fig.\ref{backyard}). In the second case they have to be approached on the angle hull again. 

Notice, that the possible paths (angel hull and semicircle) do not depend on special vertices. Therefore we can follow them via doubling until all left vertices become explored or we reach the fence. In the second case at least  one cut is not crossing $f$ and the hole is $c$--safe. It can be shown that the length of the path traveled in $\mathcal F$ is bounded by a constant times the shortest path length from $s$ to $f$, see \cite{G}.

In summary, we get the following result.

\begin{theorem}
$1$--CPEX is a $O(1)$--competitive strategy for exploring polygons with at most one colored hole and given starting point on the outer boundary.
\end{theorem}

% ---------------------------------------------------------------------------------------------------
% k-CPEX
% ---------------------------------------------------------------------------------------------------
\section{General Case: Constant Number of Colored Holes}
We extend the strategy to deal with more than one, say at most $h$, pairwise differently colored holes. The idea of $c$--safe holes can be adapted to reduce the exploration problem to a polygon with $(h-1)$ differently colored holes.

All visible holes are organized in a list $\mathcal{H}$. Each hole in $\mathcal{H}$ is marked with its current state:
\textit{discovered}, \textit{critical}, or \textit{safe}.

\subsection{Shortest Tours Around Holes}

For any $H \in {\mathcal H}$ we denote by $R_H$ the shortest tour around the hole $H$ that starts and ends in $s$.
If this tour is not unique, we choose the unique one  encircling the largest area (Fig.\ref{rch-star}(a)).
We remark that   $R_H$ can encircle other holes, too. It is also possible
that it doesn't touch $H$ at all (Fig.\ref{rch-star}(b)), and it can differ from the outer boundary of
$\mathrm{RCH} \left(H \cup \left\{ s \right\} \right)$ as well. In the situation that $R_H$ encircles
also another hole $H'$, the shortest path properties imply that either $R_{H'}=R_H$ or the region encircled by $R_{H'}$
is properly contained in the region encircled by $R_H$.

\begin{figure}[ht]
\centering
\includegraphics[scale=1]{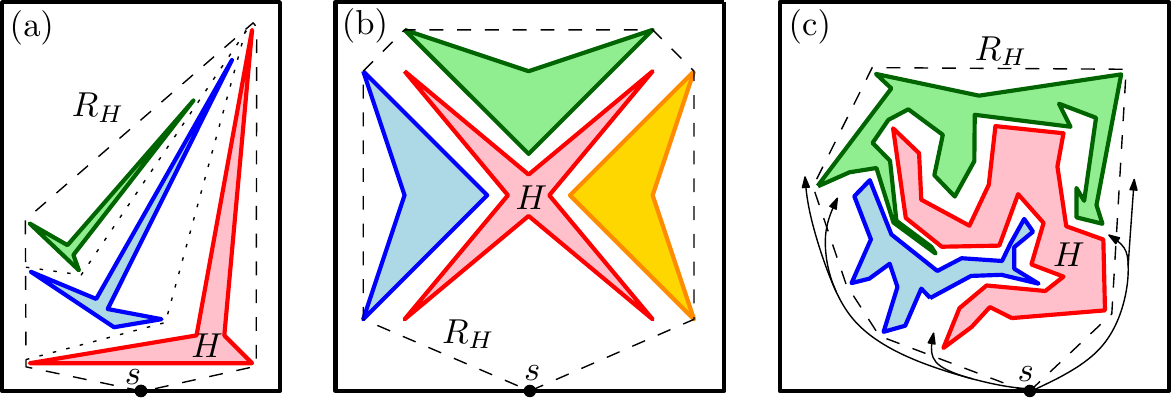}
\caption{(a) Three shortest paths encircling hole $H$ (b) $R_H$ does not touch $H$ (c) the star search approach.}
\label{rch-star}
\end{figure}

Now, the exploration of $R_H$ is more difficult because we can't predict whether it runs always through bicolored
corridors  with the color of $H$ on one side or if it also uses corridors with two other colors.
Again we start the exploration of $R_H$ cw. and ccw. around $H$, but,
whenever another hole $H'$ occurs in the search range we have to check both possibilities: $R_H$ could
run cw. or ccw. around $H'$ (Fig.\ref{rch-star}(c)). Thus, we have to replace the doubling approach from
paragraph \ref{RCH}
by a star search strategy  \cite{PY}. The bicolored corridors form the edge set $E$ of a
planar graph with $h$ faces and vertices of degree at least 3. We can conclude $|E|\leq 3h/2$ from Euler's
formula. Since each corridor will be used at most twice (from both sides) a $3h$-star search will
suffice what increases the competitive ratio for this phase by a factor of $2e\cdot 3h +1$,   \cite{PY}.

Once knowing $R_H$, we derive the lower bound $\lambda_H \leq \left|\mathcal{T}_{\mathrm{opt}}\right|$
in the same way as in Paragraph \ref{RCH}, Fig.\ref{hull}(b). Moreover, we can use all the conclusions for safe holes drawn in the $1$--hole case.
Lemma \ref{hybrid} for $c$--safe holes can be extended to $h$ holes, too.

\begin{lemma}\label{khybrid}
If a $c$--safe hole is found, the polygon can be explored with $(h-1)$--CPEX, guaranteeing a total path length $\leq (4c+2)\mathcal{C}_{h-1} \left|{\mathcal T}_{\mathrm{opt}}\right|$.
\end{lemma}

\begin{proof}
The hole $H$ has been discovered and categorized $c$--safe. Therefore we found a path $b$ connecting it with the starting point, running completely in $R_H$. We have to ensure, that
 \begin{equation}
 \label{eq_barrier}
 \left|b\right|\leq \frac{1}{2} \left|R_H\right| \enspace .
 \end{equation}
   
As mentioned before, any obstacle interfering with $R_H$ and $b$ has to be a $c$--safe hole, too. That's why a $c$--safe hole $H$ with path $b$ satisfying (\ref{eq_barrier}) can be found and the construction of Lemma \ref{hybrid} can be used.\qed
\end{proof}

Lemma \ref{khybrid} allows the recursive call of ($h-1$)--CPEX, if a $c$--safe hole is found. In that case the status of all other holes in $\mathcal H$ is reset to \textit{discovered}, because in the new derived polygon the shortest tour encircling a hole can have changed.

%%%%%%%%%%%%%%%%%%%%%%%%%%%%%%%%%%%%%%%%%%%%%%%%%%%%%%%%%%%%%%%%%%%%%%%%%%%%%%%%%%%%%%%%%%%%%%%%%%%

\subsection{The Algorithm: h--CPEX} 
For our algorithm we initialize $\mathcal{H}$ as an empty list and set $\lambda=0$. There are two basic rules $h$--CPEX will follow: 
\begin{description}
\item[(R1)]
As soon as a hole is discovered, we will classify it. Only exception: We are currently classifying another hole.
\item[(R2)]
As soon as a hole gets classified as $c$--safe, we recurse and invoke ($h-1$)--CPEX.
\end{description}
\noindent
Overall, the CPEX exploration is divided into three major steps (pseudo code can be found in  Appendix \ref{pseudocode}).

\begin{description}
\item[1. Classifying Holes] $~$\\
If no hole has been discovered yet, we apply HIKK for simple polygons until the first hole becomes visible and add it to our list $\mathcal{H}$ marked as \textit{discovered}.
Now $R_H$ has to be learned for every discovered hole $H$ with the help of the star search algorithm, visiting all possible corridors, until a point $p$ on $R_H$ is visible from both sides of the hole. If such a point is found, the strategy has to be applied another round: The shortest path could have been missed because of the malicious adversary. 
 Afterwards we compute the lower bound $\lambda_H$ and define $\lambda=\max\left(\lambda, ~ \lambda_H\right)$. If new holes are found, they are added to $\mathcal{H}$, too. 
 If $\left|R_H\right| \leq c \cdot \lambda$, the hole is safe and we apply R2. Otherwise we mark $H$ as critical. If $\lambda$ has changed, we have to check all holes  previously marked critical. They might be safe now and we can recurse, too.

\item[2. Exploring Front Yards] $~$\\
At this stage, $\mathcal{H}$ only contains critical holes. For each group $G_R$ formed by holes that have the same shortest tour $R$ surrounding them we create $\mathcal{F}_{R}$ by inserting a fence $f_R$ by cutting the polygon with the corresponding half plane through the top of $R$ (see strategy for one hole). Because the fence connects at least one hole with the outer boundary, the number of holes is decreased and we have to update list $\mathcal{H}$ (Fig.\ref{front-back}(a) and (b)). \\
Now, $(h-1)$--CPEX for $\mathcal{F}_{R}$ can be used. If its path length exceeds $\mathcal{C}_{h-1} \cdot x$, all holes in $G_R$ become $c$--safe and we recurse. Otherwise $\mathcal{F}_{R}$ is explored completely.
If a new hole is discovered, we add it to $\mathcal{H}$ and apply R1.

\begin{figure}[ht]
\centering
\includegraphics[scale=1]{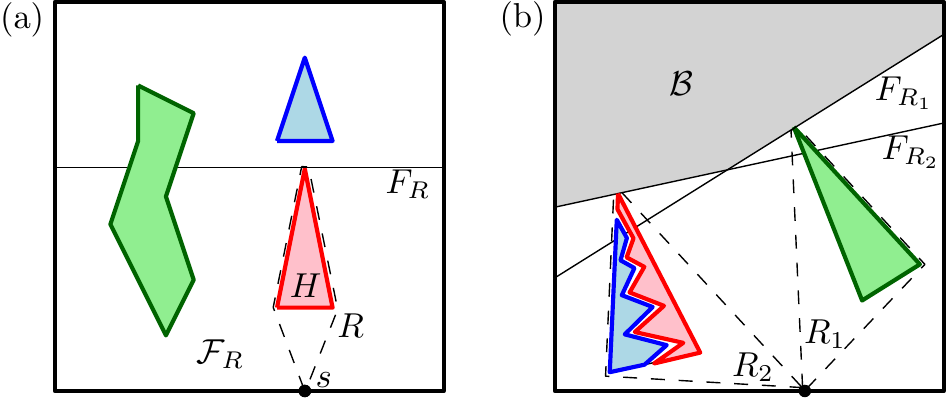}
\caption{The fence, front yard and backyard}
\label{front-back}
\end{figure}

\item[3. Exploring the Backyard] $~$\\
Finally we explore the backyard $\mathcal{B}= \mathcal{P} \setminus \bigcup_{R} \mathcal{F}_{R}$ (Fig.\ref{front-back}(b)) as described before, see \ref{sec_hybrid}. The doubling approach has to be replaced by star search again. As in step 2, if a new hole is discovered, add it to $\mathcal{H}$ and apply R1.

\end{description}

\subsection{The Competitive Factor}

\begin{theorem}
The strategy  $h$--CPEX is $(h+c_0)!$-competitive.
\end{theorem}
\begin{proof}
Recall that $\mathcal{C}_h$ denotes the competitive factor of  $h$--CPEX, for $0$--CPEX we use the
HIKK--factor $\mathcal{C}_0= 26.5$. Analyzing the different stages of  $h$--CPEX, we obtain the following
recursive estimation:
\begin{equation*}
 {\mathcal{C}_h} \leq c_1 h^2 + c_2 {\mathcal{C}_{h-1}} + h {\mathcal{C}_{h-1}} + c_3 h \enspace .
\end{equation*}

The first term comes from the classification of the $h$ holes, each using a $3h$--star search with a
$O(h)$ competitive factor. The second term comes into play whenever we have a  recursive call of
 $(h-1)$--CPEX for a safe hole. The constant $c_2=22$ stems from Lemma 3  dealing with
$5$-safe holes.
In the case that all holes are $5$-critical we have to explore at most
$h$ front yards, each implying a recursive call of   $(h-1)$--CPEX and, finally,
the exploration of the backyard that is basically an $h$--star search for groups of left and right vertices.
This estimation is obviously dominated by the second and the third term, what implies
$  {\mathcal{C}_h} \leq c_4 \cdot(h+c_2)! \cdot {\mathcal{C}_0} \leq (h+c_0)!$ for sufficiently large constants $c_4$ and $c_0$.
\qed
\end{proof} 

% ---------------------------------------------------------------------------------------------------
% Future Work
% ---------------------------------------------------------------------------------------------------
\section{Conclusion and future work}
We have addressed the problem of online exploring polygonal scenes cluttered with at most $h$ polygonal obstacles (holes). In the standard model exploring the scene includes  the subtask of
recognizing which parts of the boundary belong to holes and which edges form the outer boundary. In this paper we proposed a modified model making this subtask trivial by giving each hole a special color.

Under this assumption we could give for each $h>0$ a competitive exploration strategy. We consider this to be a major breakthrough towards settling the general conjecture from \cite{DKP-FOCS} that such competitive strategies exist in the uncolored case, too. The missing link could be a combination of star search with a HIKK--like strategy to learn which holes are there in an uncolored scene.

 Moreover, we are sure that the competitive factor can be considerably improved. We remark, that for $h \geq 2$ holes the task of exploring the polygon is no longer equivalent to exploring all edges, compare Fig.\ref{4holes}. However, this problem is not an issue for $h$--CPEX because of its recursive structure.

\begin{figure}[ht]
\centering
\includegraphics[scale=1]{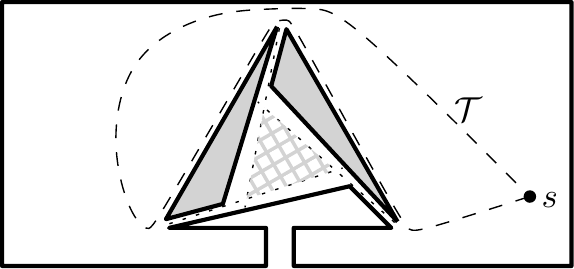}
\caption{Seeing all boundary edges does not guarantee full exploration}
\label{4holes}
\end{figure}

%---------------------------- Bibliography -------------------------------

\bibliographystyle{splncs}

% ---------------------------------------------------------------------------------------------------
% Appendix
% ---------------------------------------------------------------------------------------------------
\newpage
\begin{appendix}
\section{Appendix}
\subsection{Lower Bounds}
\label{app_orth_bound}

\begin{theorem_app}
Any deterministic online strategy $\mathcal{S}_1$ that computes valid watchman routes in orthogonal polygons with at most one hole has a competitive ratio of at least 2.
\end{theorem_app}

\begin{proof}
We confront $\mathcal{S}_1$ with a polygon composed of long thin winding corridors as indicated in Fig.\ref{lower_bound_1}. In starting point $s$ it has two choices, it can  follow corridor $L$ or $R$. Exploring them simultaneously, say by using a doubling strategy, see \cite{K}, is too expensive and because of the windings $\mathcal{S}_1$ cannot look far ahead. 
 After traveling a corridor, say $L$, a distance $d\gg 0$, $\mathcal{S}_1$ encounters a new branching region with two new corridors $R_1$ and $R_2$ pointing back. Again,  $\mathcal{S}_1$ has to decide which one eventually to follow, say it chooses $R_2$. Notice, in that moment $R_1$ has not been  completely explored and we make $R_1$ a dead end corridor by cutting it off behind the next winding.  $\mathcal{S}_1$ follows $R_2$ (which is in fact $R$) until it reaches the proximity of $s$ again. Now $\mathcal{S}_1$  will notice that it has circled the hole completely and missed the very end of $R_1$ only. To accomplish its task it has to return
to $R_1$, therefore traveling additionally  a distance of $2d$ plus twice the length of $R_1$. An optimal tour will first learn $R_1$ before it returns to $s$.
\qed
\end{proof}

\begin{figure}[ht]
\centering
\includegraphics[scale=.8]{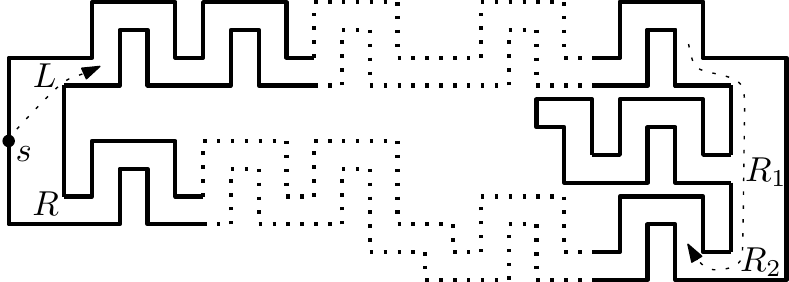}
\caption{Lower bound example, orthogonal case}
\label{lower_bound_1}
\end{figure}

Notice, that this lower bound construction does not hold for the case of a colored hole. The strategy then could identify the dead end and explore it first.

\subsection{Front Yard Exploration (Proof of Lemma \ref{P0})}
\label{app_frontyard}

We will first prove two auxiliary lemmas.

\begin{lemma}\label{alpha}
Assume that $H$ is a $5$--critical hole after the first classification, i.e. $\left|R_H \right| > 5\lambda_H$.
Let $a_i$ and $a_{i+1}$ be the two segments of $R_H$ that define $\lambda_H$, $p_i$ their
common apex and $\alpha$ the enclosed angle. Then $\alpha < \frac{\pi}{6}$.
\end{lemma}

\begin{proof}
Let  $q_l$ and $q_r$ be the two endpoints of the $\lambda_H$--path that see $p_i$ from
the left and right side. Denoting by $A=|\overline{q_l,p_i}|$ and  $B=|\overline{q_r,p_i}|$ the
distances to the apex we have $A+B+\lambda_H \geq \left|R_H \right| > 5\lambda_H $.
Consequently $A \geq 2 \lambda_H$ or  $B \geq 2 \lambda_H$ and the claim follows from the
sine rule in the triangle $\Delta(p_i,q_l,q_r)$.
\qed
\end{proof}

\begin{lemma}\label{5safe}
Assume that $H$ is a $6$--critical hole after the first classification and $f_H$ the
fence of $H$. Let $x$ denote twice the shortest path length from $s$ to $f_H$.
If $|{\mathcal T}_{\mathrm{opt}}| \geq x$, hole $H$ becomes $5$--safe.
\end{lemma}

\begin{proof}
Note that $6$-criticality implies $5$--criticality. We will use the notations from the proof of Lemma \ref{alpha}.
It is sufficient to show that $x \geq \frac{1}{5}\left|R_H\right| $ or that the length
of the shortest path from $s$ to $f_H$  is at least $0.1  \left|R_H\right|$.
Combining   $A+B+\lambda_H \geq \left|R_H\right|>6 \lambda_H $ with the triangle inequality $A+ \lambda_H \geq B$
we obtain $2A +2 \lambda_H \geq \left|R_H \right|$ and
$A \geq (\frac{1}{2}-\frac{1}{6}) \left|R_H\right|= \frac{1}{3} \left|R_H\right|$.
From Lemma \ref{alpha} we know that the shortest path length from $q_l$ to
$f_H$ is at least $A \cos{\frac{\pi}{6}}\geq 0.86 A$. Since the distance from $s$ to $q_l$
in ${\mathcal P}$ is at most $\lambda_H \leq \frac{1}{6} \left|R_H\right|$ we end with
\begin{equation*}
\frac{x}{2} \geq 0.86A-\lambda_H \geq \left(\frac{0.86}{3}-\frac{1}{6}\right) \left|R_H\right| =0.12 \left|R_H\right| \enspace .
\end{equation*}
\qed
\end{proof}

Now we can prove Lemma \ref{P0} in the following more general form:
\begin{lemma}\label{frontyard}
Let $\mathcal P$ be a polygon with at most $h$ holes,
$H$  a $6$--critical hole after the first classification and $x$ be twice the shortest path
length from $s$ to the fence $f_H$ in $\mathcal P$. Starting $(h-1)$--CPEX in the front yard
$\mathcal F_H$ we have: $\mathcal F_H$ will be explored with path length $\leq \mathcal C_{h-1} \cdot x$
or $H$ is $5$--safe.
\end{lemma}

\begin{proof}
The front yard exploration by $(h-1)$--CPEX will stop if either \\
$(1)$ the exploration  path length $l$ reaches ${\mathcal C}_{h-1}x$ and ${\mathcal F}_H$ is still unexplored
or if \\
$(2)$  ${\mathcal F}_H$ gets explored with path length $l \leq {\mathcal C}_{h-1}x$.

We will prove, that this
procedure is ${\mathcal C}_{h-1}$--competitve in both cases, i.e.
$l \leq {\mathcal C}_{h-1}\left|{\mathcal T}_{\mathrm{opt}}\right|$, and that in case $(1)$ the hole
$H$ becomes $5$--safe. We remind that ${\mathcal T}_{\mathrm{opt}}$ is the optimal exploration tour for the whole
polygon whereas ${\mathcal T}_{\mathrm{opt}}(\mathcal F)$ denotes an optimal tour for the front yard.
Now we combine our two cases with another case distiction: \\
$(a)$ $|{\mathcal T}_{\mathrm{opt}}| \leq x$
and $(b)$  $|{\mathcal T}_{\mathrm{opt}}| > x$.

Remark that in case (a) ${\mathcal T}_{\mathrm{opt}}$ does not
leave $\mathcal F$ and, consequently, $\left|{\mathcal T}_{\mathrm{opt}}(\mathcal F) \right| \leq \left|{\mathcal T}_{\mathrm{opt}}\right|$.
In case (b) the hole becomes $5$-safe by Lemma \ref{5safe}.

\begin{description}
\item[(1.a)] Since the exploration is not finished and   $(h-1)$--CPEX is ${\mathcal C}_{h-1}$--competitive we have
$l = {\mathcal C}_{h-1}x <  {\mathcal C}_{h-1} \left|{\mathcal T}_{\mathrm{opt}}(\mathcal F) \right| \stackrel{(a)}{\leq} {\mathcal C}_{h-1}
\left|{\mathcal T}_{\mathrm{opt}} \right|$ and  $x <\left|{\mathcal T}_{\mathrm{opt}} \right|$ by canceling  ${\mathcal C}_{h-1}$.
Hence, $H$  becomes $5$-safe by Lemma \ref{5safe}.

\item[(1.b)]  $l = {\mathcal C}_{h-1}x \stackrel{(b)}{<}  {\mathcal C}_{h-1} \left|{\mathcal T}_{\mathrm{opt}} \right|$ and $H$ is $5$--safe.

\item[(2.a)]  $l \stackrel{(2)}{\leq} {\mathcal C}_{h-1}\left|{\mathcal T}_{\mathrm{opt}}(\mathcal F) \right|
\stackrel{(a)}{\leq} {\mathcal C}_{h-1}\left|{\mathcal T}_{\mathrm{opt}} \right|$.

\item[(2.b)]  $l \stackrel{(2)}{\leq} {\mathcal C}_{h-1}x  \stackrel{(b)}{\leq} {\mathcal C}_{h-1}\left|{\mathcal T}_{\mathrm{opt}} \right|$.
\end{description}
\qed
\end{proof} 

\subsection{Backyard Exploration}
\label{app_background}

Here we describe the construction, application, and analysis of the angle hull for exploring left vertices in $\mathcal B$ on the left side of the hole.

\begin{definition}
Let $\mathcal D$ be a simple polygon contained in another simple polygon $\mathcal P$. The \textit{angle hull} $\mathcal{AH}(\mathcal D)$ of $\mathcal D$ consists of all points in $\mathcal P$ that can see two points of $\mathcal D$ at an angle of $90^\circ$ \cite{HIKK}.
\end{definition}

Let $p_l$ be the leftmost point of $R$ related to $f$ and $p'_l$ its projection on $f$ (Fig.\ref{angle_hull}). Each shortest path form $s$ to a cut (belonging to a remaining left reflex vertex in $\mathcal B$ and crossing $f$ on the left) touches the cut in the area between the hole and the straight line through $p_l$ and $p'_l$. 
We try to explore all of these remaining vertices by following the shortest path from $s$ to $p_l$ and approaching $f$ on the angle hull $\mathcal{AH}$ of the shortest path from $p_l$ to $p_{i_0}$ afterwards. Notice, that this path $\Pi$ does not depend on any vertex in $\mathcal B$.

\begin{figure}[ht]
\centering
\includegraphics[scale=.8]{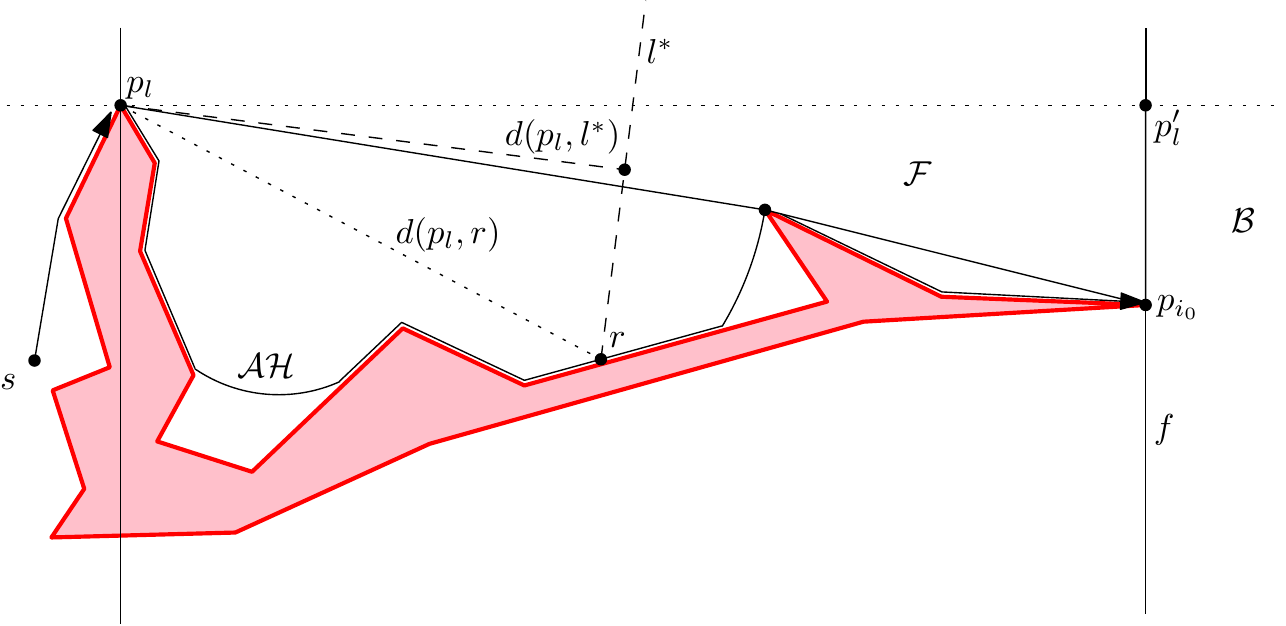}
\caption{The angle hull.}
\label{angle_hull}
\end{figure}

\begin{lemma}
\label{backyard_angle_hull}
If a left reflex vertex $l$ becomes explored by following the path $\Pi$, then the length of the traveled path is bounded by four times the length of the shortest path from $s$ to the corresponding cut $l^*$.
\end{lemma}

\begin{proof}
The angle hull $\mathcal{AH}$ intersects $l^*$ in point $r$. Its length from $p_l$ to $r$
is bounded by twice the length $d(p_l,r)$ of the shortest path connecting those both points \cite{HIKK}.
Using elementary trigonometric reasoning one can show that the length of this path is again bounded by twice the distance $d(p_l,l^*)$ of $p_l$ to the cut $l^*$. Together with the optimal path from $s$ to $p_l$ the claim follows. 
\qed
\end{proof}

Notice, that the optimal tour does not has to explore a vertex from the left side of the hole. But in this case the strategy for the hole's right side is competitive and the path length traveled on the left side is bounded by the doubling approach. 

Furthermore the strategy does not depend on the location of the vertices in $\mathcal B$. If the fence $f$ is touched on both sides we can be sure that there has to be a cut not crossing $f$ and the hole becomes safe. Otherwise the polygon is explored before and Lemma \ref{backyard_angle_hull} grants competitiveness of the exploration.

\subsection{h-CPEX Pseudocode}
\label{pseudocode}

% k-CPEX ---------------------------------------------------------------------------------------------------
\begin{algorithm}
	\caption{$h$-CPEX}
	\label{epch}
	\begin{algorithmic}[1]
	\Procedure{CPEX}{$\mathcal{P}$, $s$, $\mathcal{H}$, $\lambda$, $h$}
		\While{$\mathcal{H}$ is empty}
				\State apply HIKK for simple polygons until first hole is found
		\EndWhile
		\ForAll{$H \in \mathcal{H}$ marked \textit{discovered}}
				\State learn shortest path $R_H$, compute $\lambda_H$, and add new holes to $\mathcal{H}$
				\State $\lambda=\max\left(\lambda, ~ \lambda_H\right)$
				\If {$\left|R_H\right| \leq c \cdot \lambda$}
						\Comment{Is $H$ $c$--safe?}
						\State create $\mathcal{P}'$ by inserting barrier $b$ and call CPEX($\mathcal{P}'$, $s$, $\mathcal{H}$, $\lambda$, $h-1$) 
							
				\Else
						\State mark $H$ as critical
				\EndIf
				\ForAll {$H \in \mathcal{H}$ marked \textit{critical}}
						\If {$\left|R_H\right| \leq c \cdot \lambda$}
							\Comment{Update status}
							\State create $\mathcal{P}'$ by inserting barrier $b$ and call CPEX($\mathcal{P}'$, $s$, $\mathcal{H}$, $\lambda$, $h-1$)
								
						\EndIf
				\EndFor
		\EndFor
		\While{no new hole becomes visible}
				\ForAll{$H \in \mathcal{H}$ marked \textit{critical}}
						\State create $\mathcal{F}_{H}$ by inserting fence line $f_H$ and update list $\mathcal{H}$
						\While{tour length $\leq x \cdot \mathcal C_{h-1}$}
								\State CPEX($\mathcal{F}_{H}$, $s$, $\mathcal{H}$, $\lambda$, $h-1$)
						\EndWhile
						\If{tour length $\leq  x \cdot \mathcal C_{h-1}$}
								\Comment{Is $H$ $c$--safe?}
								\State create $\mathcal{P}'$ by inserting barrier $b$ and call CPEX($\mathcal{P}'$, $s$, $\mathcal{H}$, $\lambda$, $h-1$)
						\EndIf
				\EndFor
				\Comment{All front yards are explored}
				\State $\mathcal{B}= \mathcal{P} \setminus \bigcup \mathcal{F}_{H}$ and explore it with star search 
		\EndWhile
		\State add new hole to $\mathcal{H}$ and call CPEX($\mathcal{P}$, $s$, $\mathcal{H}$, $\lambda$, $h$)
		\Comment{Restart to classify new hole}
	\EndProcedure
	\end{algorithmic}
\end{algorithm}

\end{appendix}

\end{document}